\documentclass[journal]{IEEEtran} 
\usepackage{bm}
\usepackage{epsf}
\usepackage{times}
\usepackage{float}
\usepackage{graphicx}
\usepackage{amsmath,amssymb,amsthm}
\usepackage{color}
\usepackage{cite}
\usepackage{subcaption}
\usepackage{algorithm}
\usepackage{algorithmic}

\usepackage{float} 

\newtheorem{remark}{Remark}
\newtheorem{proposition}{Proposition}

\graphicspath{{subPert/}}

\def \bx{\bm x}
\def \bQ{\bm Q}
\def \bs{\bm s}

\def \bX{\bm X}
\def \by{\bm y}
\def \bz{\bm z}
\def \bv{\bm v}
\def \bC{\bm C}
\def \bPC{\bm {PC}}
\def \bP{\bm P}
\def \bQ{\bm Q}
\def \bB{\bm B}
\def \bI{\bm I}
\def \blambda{\bm \lambda}

\begin{document}
\raggedbottom
\begin{sloppy}
\addtolength{\abovedisplayskip}{-1.0mm}
\addtolength{\belowdisplayskip}{-1.0mm}
\title{\huge Privacy-Preserving Distributed Optimization via Subspace Perturbation:
 A General Framework}

\author{Qiongxiu Li,~\IEEEmembership{Student Member,~IEEE}, Richard Heusdens
and Mads~Gr\ae sb\o ll~Christensen,~\IEEEmembership{Senior Member,~IEEE}
\thanks{Q. Li and M. G. Christensen are with the Audio Analysis Lab, CREATE, Aalborg University, 9000 Aalborg, Denmark  (emails: \{qili,mgc\}@create.aau.dk).}
\thanks{R. Heusdens is with the Department of Microelectronics, Faculty of Electrical Engineering, Mathematics and Computer Science, Delft University of Technology, Delft 2628, CD,
The Netherlands, and also with Netherlands Defence Academy, Kraanstraat 4, The Netherlands (email: R.Heusdens@tudelft.nl).}
}


\maketitle

\begin{abstract}
As the modern world becomes increasingly digitized and interconnected, distributed signal processing has proven to be effective in processing its large volume of data. However, a main challenge limiting the broad use of distributed signal processing techniques is the issue of privacy in handling sensitive data. To address this privacy issue, we propose a novel yet general subspace perturbation method for privacy-preserving distributed optimization, which allows each node to obtain the desired solution while protecting its private data. In particular, we show that the dual variables introduced in each distributed optimizer will not converge in a certain subspace determined by the graph topology. Additionally, the optimization variable is ensured to converge to the desired solution, because it is orthogonal to this non-convergent subspace. We therefore propose to insert noise in the non-convergent subspace through the dual variable such that the private data are protected, and the accuracy of the desired solution is completely unaffected. Moreover, the proposed method is shown to be secure under two widely-used adversary models: passive and eavesdropping. Furthermore, we consider several distributed optimizers such as ADMM and PDMM to demonstrate the general applicability of the proposed method. Finally, we test the performance through a set of applications. Numerical tests indicate that the proposed method is superior to existing methods in terms of several parameters like estimated accuracy, privacy level, communication cost and convergence rate.
\end{abstract}


\section{Introduction}
In a world of interconnected and digitized systems, new and innovative signal processing tools are needed to take advantage of the sheer scale of information/data. Such systems are often characterized by ``big data''. Another central aspect of such systems is their distributed nature, in which the data are usually located in different computational units that form a network. In contrast to the traditional centralized systems, in which all the data must be firstly collected from different units and then processed at a central server, distributed signal processing circumvents this limitation by utilizing the network nature. That is, instead of relying on a single centralized coordination, each node/unit is able to collect information from its neighbors and also to conduct computation on a subset of the overall networked data. This distributed processing has many advantages, such as allowing for flexible scalability of the number of nodes and robustness to dynamical changes in the graph topology. 
Currently, the computational unit/node in distributed systems is usually limited in resources, as tablets and phones become the primary computing devices used by many people \cite{anderson2015technology,poushter2016smartphone}. These devices often contain sensors that can use wireless communication to form so-called ad-hoc networks. Therefore, these devices can collaborate in solving problems by sharing computational resources and sensor data. However, the information collected from sensors such as GPS, cameras and microphones often includes personal data, thus posing a major concern, because such data are private in nature. 

There has been a considerable growth of optimization techniques in the field of distributed signal processing, as many traditional signal processing problems in distributed systems can be equivalently formed as convex optimization problems. Owing to the general applicability and flexibility of distributed optimization, optimization has emerged in a wide range of applications such as acoustic signal processing \cite{sherson2016distributed,koutrouvelis2018low}, control theory~\cite{zhu2011distributed} and image processing~\cite{zibulevsky2010l1}. Typically, the paradigm of distributed optimization is to separate the global objective function over the network into several local objective functions, which can be solved for each node through exchanging data only with its neighborhood. This data exchange is a major concern regarding privacy, because the exchanged data usually contain sensitive information, and traditional distributed optimization schemes do not address this privacy issue. Therefore, how to design a distributed optimizer for processing sensitive data, is a challenge to be overcome in the field. 

\subsection{Related works}
To address the privacy issues in distributed optimization, the literature has mainly used techniques from differential privacy \cite{dwork2006,dwork2006calibrating} and secure multiparty computation (SMPC) \cite{Cramer2015}. Differential privacy is one of the most commonly used non-cryptographic techniques for privacy preservation, because it is computationally lightweight, and it also uses a strict privacy metric to quantify that the posterior guess of the private data is only slightly better than the prior (quantified by a small positive number $\epsilon$). This method of protecting private data has been applied in \cite{huang2015differentially,han2016differentially,nozari2018differentially,zhang2016dynamic,zhang2018recycled,zhang2018improving,xiong2020privacy} through carefully adding noise to the exchanged states or objective functions. However, this noise insertion mechanism involves an inherent trade-off between the privacy and the accuracy of the optimization outputs. 
Additionally, some approaches \cite{huang2019dp,hale2015differentially,hale2018cloud} have applied differential privacy with the help of a trusted third party (TTP) like a server/cloud. However, requiring a TTP for coordination makes the protocol not completely decentralized (i.e., peer-to-peer setting). Consequently, it thus hinders use in many applications in which centralized coordinations are unavailable. 

SMPC, in contrast, has been widely used in distributed processing, because it provides cryptographic techniques to ensure privacy in a distributed network. More specifically, it aims to compute the result of a function of a number of parties' private data while protecting each party's private data from being revealed to others. Examples of how to preserve privacy by using SMPC have been applied in \cite{xu2015secure,freris2016distributed,shoukry2016privacy,wang2011secure}, in which partially homomorphic encryption (PHE) \cite{paillier1999public} was used to conduct computations in the encrypted domain. However, PHE requires the assistance of a TTP and thus cannot be directly applied in a fully decentralized setting. 
Additionally, although PHE is more computationally lightweight than other encryption techniques, such as fully homomorphic encryption \cite{gentry2009fully} and Yao's garbled circuit \cite{yao1982protocols,yao1986generate}, it is more computationally demanding than the noise insertion techniques, such as differential privacy, because it relies on the computational hardness assumption. To alleviate the bottleneck of computational complexity, another technique in SMPC, called secret sharing \cite{damgaard2012multiparty}, has become a popular alternative for distributed processing, because its computational cost is comparable to that of differential privacy. It has been applied in \cite{tjell2020privacy} to preserve privacy by splitting sensitive data into pieces and sending them to the so-called computing parties afterward. However, secret sharing generally is expensive in terms of communication cost, because it requires multiple communication rounds for each splitting process. 

\subsection{Paper contributions}
The main contribution of this paper is that we propose a novel subspace perturbation method, which circumvents the limitations of both the differential privacy and SMPC approaches in the context of distributed optimization. We propose to insert noise in the subspace such that not only the private data is protected from being revealed to others but also the accuracy of results is not affected. The proposed subspace perturbation method has several attractive properties:
\begin{itemize}
  \item Compared to differential privacy based approaches, the proposed approach is ensured to converge to the optimum results without compromising privacy. Additionally, it is defined in a completely decentralized setting, because no central aggregator is required.
  \item In contrast to SMPC based approaches, the proposed approach is efficient in both computation and communication. Because it does not require complex encryption functions (such as those involved in PHE), it does not have high communication costs (such as those required in the secret sharing approaches).
  \item The proposed subspace perturbation method is generally applicable to many distributed optimization algorithms like ADMM, PDMM or the dual ascent method.
  \item The convergence rate of the proposed method is invariant with respect to the amount of inserted noise and thus to the privacy level. 
\end{itemize}
We published preliminary results in \cite{JaneICASSP2020} where the main concept of subspace perturbation was introduced using PDMM with one specific application, and here we give more complete analysis of the proposed subspace perturbation and further generalize it into other optimizers and various applications.    
\subsection{Outline and notation}
The remainder of this paper is organized as follows: Section \ref{section: preliminary} reviews distributed convex optimization and some important concepts for privacy preservation. Section \ref{section:problemDef} defines the problem to be solved and provides qualitative metrics to evaluate the performance. Section \ref{section:pdmm} introduces the primal-dual method of multipliers (PDMM), explaining its key properties used in the proposed approach. Section \ref{section:prop} introduces the proposed subspace perturbation method based on the PDMM. Section \ref{section:genFramework} shows the general applicability of the proposed method by considering other types of distributed optimizers, such as ADMM and the dual ascent method. In Section \ref{section:applications} the proposed approach is applied to a wide range of applications including distributed average consensus, distributed least squares and distributed LASSO. Section \ref{section:numerical} demonstrates the numerical results for each application and compares the proposed method with existing approaches. 

The following notations are used in this paper. Lowercase letters $(x)$, lowercase boldface letters $(\bx)$,  uppercase boldface letters $(\bX)$, overlined uppercase letters $(\bar{X})$ and calligraphic letters $(\mathcal{X})
$ denote scalars, vectors, matrices, subspaces and sets, respectively. An uppercase letter $(X)$ denotes the random variable of its lowercase argument, which means that the lowercase letter $x$ is assumed to be a realization of random variable $X$. 
$\operatorname { null }\{\cdot\} \text { and } \operatorname{span}\{\cdot\}
$ denote the nullspace and span of their argument, respectively.
$(\bX)^\dagger$ and $(\bX)^{\top}$ denote the Moore-Penrose pseudo inverse and transpose of $\bX$, respectively. $\bx_{i}$ denotes the $i$-th entry of the vector $\bx$, and $X_{ij}$ denotes the $(i,j)$-th entry of the matrix $X$.
$\boldsymbol{0}$, $\boldsymbol{1}$ and $ \bm I$ denote the vectors with all zeros and all ones, and the identity matrix of appropriate size, respectively.
\section{Fundamentals}\label{section: preliminary}
In this section, we review the fundamentals and some important concepts related to privacy preservation. We first review the distributed convex optimization and highlight its privacy concerns. Then we describe the adversary models that will be addressed later in this paper. 
\subsection{Distributed convex optimization }\label{subsec:convex}
A distributed network is usually modeled as a graph $\mathcal{G}=(\mathcal{N},\mathcal{E})$, where $\mathcal{N}={\{1,2,...,n}\}$ is the set of nodes, and $\mathcal{E}\subseteq \mathcal{N}\times \mathcal{N}$ is the set of edges. Let $n=|\mathcal{N}|$ and $m=|\mathcal{E}|$ denote the numbers of nodes and edges, respectively. $\mathcal{N}_{i}=\{j\,| \,{(i,j)\in \mathcal{E}\}}$ denotes the neighborhood of node $i$, and $d_{i}=|\mathcal{N}_{i}| $ denotes the degree of node $i $.

Let $\boldsymbol{x_i}\in \mathbb{R}^{u_i}$ and $\bs_{i}\in \mathbb{R}^{p_i}$ denote the local optimization variable and input/measurement at node $i$, respectively. 
A standard constrained convex optimization problem over the network can then be expressed as
\begin{equation} \label{eq.setup}
\begin{array}{ll} {\displaystyle \min_{\bx_i}} & {{\displaystyle\sum_{i \in \mathcal{N}}} f_{i}\left(\bx_{i},\bs_i\right)} 
\\ {\text {s.t.}} & {\bB_{i|j}\bx_{i}+\bB_{j|i}\bx_{j}=\bm b_{i,j} ~ \forall (i,j)\in \mathcal{E}}
\end{array}
\end{equation}
where $ f_{i}: \mathbb{R}^{u_i} \times \mathbb{R}^{p_i}\mapsto \mathbb{R} \cup\{\infty\}$ denote the local objective function at node $i$, which we assume to be convex for all nodes $i \in \mathcal{N}$. Additionally, let $v_{i,j}$ denote the dimension of constraints at each edge $(i,j)\in \mathcal{E}$,  $\bB_{i|j} \in \mathbb{R}^{v_{i,j}\times u_i}$, $\bm b_{i,j}\in \mathbb{R}^{v_{i,j}}$ are defined for the constraints. Note that we distinct the subscripts $i|j$ and $i,j$, where the former is  a directed identifier that denotes the directed edge from node $i$ to $j$ and the later $i,j$ is an undirected identifier. Stacking all the local information and let $N_n=\sum_{i \in \mathcal{N}} u_i$, $P_n=\sum_{i \in \mathcal{N}} p_{i}$, $M_m=\sum_{(i,j) \in \mathcal{E}} v_{i,j}$, we can compactly express \eqref{eq.setup} as
\begin{equation} \label{eq.compact}
\begin{array}{ll} {\displaystyle \min_{\bx}} & {f(\bx,\bs)} \\ {\text {s.t.}} & {\bB\bx=\bm b}\end{array},
\end{equation}
where $f: \mathbb{R}^{N_n}\times \mathbb{R}^{P_n} \mapsto \mathbb{R} \cup\{\infty\}$,  $\bs\in \mathbb{R}^{P_n}$, $\bx\in \mathbb{R}^{N_n}$, $\bB\in \mathbb{R}^{M_m\times N_n}$, $\bm b \in \mathbb{R}^{M_m}$. For simplicity, we assume the dimension of $\bx_i$ of all nodes are the same and set it as $u$, i.e., $u=u_i, \forall i\in \mathcal{N}$ and let all $\bB_{i|j}$ be square matrices, i.e., $v_{i,j}=u_i=u, \forall (i,j)\in \mathcal{E}$. we thus have $M_m=m\times u$ and $N_{n}=n\times u$.  We further define matrix $\bB$ based on the incidence matrix of the graph: $\bB_{i|j}=\bI$ , $\bB_{j|i}=-\bI $ if and only if $(i,j)\in \mathcal{E}$ and $i<j$ and $\bB_{i|j}=-\bI$ , $\bB_{j|i}=\bI $ if and only if $ (i,j)\in \mathcal{E}$ and $i>j$. Note that $\bB$ reduces to the incidence matrix if $u=1$. 
To simplify notation, we will in what follows drop the $\bs$-dependency in the objective functions and simply write $f(\bx)$ and $f_i(\bx_i)$.

To solve the above problem without any centralized coordination, several distributed optimizers have been proposed, such as ADMM \cite{boyd2011distributed} and PDMM \cite{zhang2018distributed, sherson2018derivation}, to iteratively solve the problem by communicating only with the local neighborhood. That is, at each iteration (denoted by index $k$), each node $i$ updates its optimization variable $\bx_{i}^{(k)}$ by exchanging data only with its neighbors. The goal of distributed optimizers is to design certain updating functions to ensure that $\bx_{i}^{(k)} \rightarrow \bx_{i}^{*}$, where $\bx_{i}^{*}$ denotes the optimum solution for node $i$. Generally, these updating functions are functions of the input $\bs_{i}$. 

\subsection{Privacy concerns}
As mentioned in the introduction, the sensor data captured by an individual's device are usually private in nature. For example, health conditions can be revealed by voice signals \cite{alavijeh2019quality,poorjam2019automatic}, and activities of householders can be revealed by power consumption data \cite{giaconi2018privacy}. Therefore, we are able to identity the local input/measurement $\bs_{i}$ held by each node $i$ as the private data to be protected in the context of distributed optimization. 
Recall that after each iteration, node $i$ will send the updated optimization variable $\boldsymbol{x_i}^{(k+1)}$ to all of its neighbors. Since this variable is computed through a function having the private data $\bs_{i}$ as input, the revealed $\bx_{i}^{(k+1)}$ leaks information about $\bs_{i}$. This privacy concern, however, has not been addressed in existing distributed optimizers. 
Therefore, in this paper, we attempt to investigate this privacy issue and propose a general solution to achieve privacy-preserving distributed optimization. 

\subsection{Adversary model}
When designing a privacy-preserving algorithm, it is important to determine the adversary model that qualifies its robustness under various types of security attack. By colluding with a number of nodes, the adversary aims to conduct certain malicious behaviors, such as learning private data or manipulating the function result to be incorrect. These colluding and non-colluding nodes are referred to as corrupted nodes and honest nodes, respectively. Most of the literature has considered only a passive (also called honest-but-curious or semi-honest) adversary, where the corrupted nodes are assumed to  follow the instructions of the designed protocol, but are curious about the honest nodes' private data, i.e., the local measurements $\bs_{i}$. 
Another common adversary is the eavesdropping adversary, which can be internal or external with respect to the network and also aims to infer the private data of the honest nodes. The eavesdropping adversary in the context of privacy-preserving distributed optimization is relatively unexplored. In fact, many SMPC based approaches, such as secret sharing \cite{li2019privacyS,tjell2019privacy,tjell2020privacy}, assume that all messages are transmitted through securely encrypted channels \cite{dolev1993perfectly}, such that the communication channels cannot be eavesdropped upon. However, channel encryption is computationally demanding and is therefore very expensive for iterative algorithms, such as those used here, because they require use of communication channels between nodes many times. In this paper, we design the privacy-preserving distributed optimizers in an efficient way, such that the channel encryption needs to be used only once. 
\section{Problem definition}\label{section:problemDef}
Given the above-mentioned fundamentals, we thus conclude that the goal of privacy-preserving distributed convex optimization is to jointly optimize the constrained convex problem on the basis of the private data of each node while protecting the private data from being revealed under defined adversary models. More specifically, there are two requirements that should be satisfied simultaneously:
\begin{itemize}
\item[1)] Output correctness: at the end of the algorithm, each node $i$ obtains its optimum solution $\bx_{i}^{*}$ and its correct function result $f_{i}(\bx_{i}^{*})$, which implies that the global function result $f(\bx^{*})$ has been also achieved.  
\item[2)] Individual privacy: throughout the execution of the algorithm, the private data $\bs_{i}$ held by each honest node is protected against both passive and eavesdropping adversaries; except for the information that can be directly inferred from the knowledge of the function results and the private data of the corrupted nodes (in section \ref{subsec:lowerBound} we will explain this in detail).
\end{itemize}
To quantify the above requirements, two metrics must be defined. 
\subsection{Output correctness metric}
For each node $i$, achieving the optimum solution $\bx_{i}^{*}$ implies obtaining the correct function output $f_{i}(\bx_{i}^{*})$ as well. To measure the output correctness for the whole network in terms of the amount of communication, we thus use the mean squared error $\|\bx^{(k)}-\bx^{*}\|_2^{2}$ over all the nodes as a function of number of transmission: one transmission denotes that one message package is transmitted from one node to another. 

\subsection{Individual privacy metric} \label{subsec:privacyMeasure}
We now define how to qualitatively measure the individual privacy. Let $\mathcal{N}_{c}$ and $\mathcal{N}_{h}$ denote the set of corrupted and honest nodes, respectively.  Given that distributed optimizers usually require an iterative process, the following criteria are considered in evaluating individual privacy: 

\subsubsection{Local information leakage in each iteration} 
Without loss of generality, we assume that the corrupted nodes are attempting to infer the private data $\bs_{i}$ of honest node $i \in \mathcal{N}_{h}$. In addition, let $v_{i}^{(k)}$ denote the information collected by the corrupted nodes at iteration $k$ to infer information about the private data $\bs_{i}$.
The local information leakage  of node $i$ at iteration $k$ is measured by 
\begin{align}
 \ell^{(k)}_{i}&= I(S_{i};V_{i}^{(k)})= h(S_{i}) - h(S_{i}\, |\,V_{i}^{(k)}),
\end{align}
where $I(\cdot\, ;\cdot)$ denotes mutual information \cite{cover2012elements} and $h(S_{i})$ denotes the differential entropy of $S_{i}$, assume it exists\footnote{For the case that $S_{i}$ is a discrete random variable, the condition is given by Shannon entropy $H(S_{i})$}.

\subsubsection{Cumulative information leakage across all iterations} At the end of the algorithm, the corrupted nodes can combine all the information collected over all the iterations to infer the private data $\bs_{i}$. Let $\mathcal{V}_{i}$ denote the set of cumulative information $\mathcal{V}_{i}=\{v_{i}^{(k)}\}_{k\in \mathcal{K}}$, where $\mathcal{K}=\{0,1,\ldots,K\}$ denotes the set of iterations. The cumulative information leakage becomes
\begin{align}
 \ell_{i}=I(S_{i};V_{i}^{(0)},\ldots,V_{i}^{(K)}).
\end{align}
Of note, $\ell^{(k)}_{i}\leq \ell_{i}$ as $\forall k, \, v_{i}^{(k)} \in \mathcal{V}_{i}$.

\subsubsection{Lower bound of information leakage}\label{subsec:lowerBound}
When defining the individual privacy, we explicitly exclude the information that can be deduced from the function output and the private data of corrupted nodes, because each node will eventually obtain its output from the algorithm, and in some cases, this output may contain certain information regarding the private data held by the individual honest node. To explain this scenario more explicitly, take the distributed average consensus as an example. A group of $n$ people would like to compute their average salary, denoted by ${s_{\mathrm{ave}}}$, while keeping each person's salary $\bs_{i}$ unknown to the others. If the average result is accurate, the salary sum of the honest people can always be computed by $\sum_{i \in \mathcal{N}_{h}}{\bs_{i}}=n\times {s_{\mathrm{ave}}}-\sum_{i \in \mathcal{N}_{c}}{\bs_{i}}$ assuming the adversary knows $n$, regardless of the underlying algorithms. With the mutual information metric, the salary sum will leak $I(S_{i};\sum_{j \in \mathcal{N}_{h}}S_{j})$ amount of information about the salary of the honest node $i$. Provided that this information leakage is unavoidable, we therefore refer to it as the lower bound of information leakage. A privacy-preserving algorithm is considered perfect (or achieves perfect security) as long as it reveals no more information than this lower bound. 

\textbf{Sufficient conditions for perfect security.}
Let the mutual information $I^{*}_{i}$ denote the lower bound of information leakage for node $i$. We conclude that perfect security can be achieved if both the local and accumulative information leakage do not exceed the lower bound. That is, 
\begin{align}\label{eq.perPriv}
    \forall i\in \mathcal{N}, k \in \mathcal{K}\,:\, \ell^{(k)}_{i}\leq \ell_{i}\leq I^{*}_{i}.
\end{align}

\section{Primal-dual method of multipliers}\label{section:pdmm}
Among possible optimizers, we use PDMM to show the proposed subspace perturbation because of its broadcasting property (see Section \ref{subsec:bPDMM}), which allows for simplification of the individual privacy analysis. Moreover, it yields further insight into the constructed subspace. In Section \ref{section:genFramework} we will consider other distributed optimizers. Here we first provide a review of the fundamentals of the PDMM and then introduce its main properties, which will be used later in the proposed approach.
\subsection{Fundamentals}
PDMM is an instance of Peaceman-Rachford splitting of the extended dual problem (refer to \cite{sherson2018derivation} for details). It is an alternative distributed optimization tool to ADMM for solving constrained convex optimization problems and is often characterized by a faster convergence rate \cite{zhang2018distributed}. 
For the distributed optimization problem stated in \eqref{eq.setup}, the extended augmented Lagrangian of PDMM is given by
\begin{align}\label{eq.pdmmSetup}
 L(\bx,\blambda) = f(\bx)+ (\bP\blambda^{(k)})^{\top}\!\!\bC\bx + \frac{c}{2}\|\bC\bx + \bPC\bx^{(k)}-2\boldsymbol{d}\|_2^2, 
\end{align}
and the updating equations of PDMM are given by
\begin{align}
\bx^{(k+1)} &= \arg\min_{\bx} L\left(\bx,\bx^{(k)},\blambda^{(k)} \right), \label{eq.pdmmxUpdate} \\
\blambda^{(k+1)} &= \bP\blambda^{(k)} + c(\bC\bx^{(k+1)} + \bPC\bx^{(k)}-2\boldsymbol{d}), \label{eq.pdmmlUpdate}
\end{align}
where $\bP \in \mathbb{R}^{2 M_m \times 2 M_m}$ is a symmetric permutation matrix exchanging the first $M_m$ with the last $M_m$ rows, and $\boldsymbol{d}=\frac{1}{2}[\boldsymbol{b}^{\top} ~ \boldsymbol{b}^{\top}] \in \mathbb{R}^{2 M_m}$, $c>0$ is a constant controlling the convergence rate. 
$\blambda^{(k)} \in \mathbb{R}^{2 M_m}$ denotes the dual variable at iteration $k$, introduced for controlling the constraints. Each edge $(i,j)\in \mathcal{E}$ is related to two dual variables $\blambda_{i|j},\blambda_{j|i}\in \mathbb{R}^{u}$, controlled by node $i$ and $j$, respectively. 
Additionally, $C \in \mathbb{R}^{2 M_m \times N_n}$ is a matrix related to $\bB$:  $\bC=[\bB_{+}^{\top},\bB_{-}^{\top}]^{\top}$, where $\bB_{+}$ and $\bB_{-}$ are the matrices containing only the positive and negative entries of $\bB$, respectively. 
Of note, $\bC + \bPC = [\bB^{\top} \, \bB^{\top}]^{\top}$ and $\forall (i,j)\in \mathcal{E}: \blambda_{j|i} = \left(\bP\blambda\right)_{i|j}$. 

\subsection{Broadcast PDMM}\label{subsec:bPDMM}
On the basis of \eqref{eq.pdmmxUpdate}, the local updating function at each node $i$ is given by 
\begin{align}
 \bx_{i}^{(k+1)} = &\arg\min_{\bx_{i}} \left(
 f_{i}(\bx_{i})+ \sum_{j \in \mathcal{N}_i} {\blambda_{j|i}^{(k)^{\top}}}\!\!\bB_{i|j}\bx_i \right. \nonumber
 \\ &\hspace{5mm}+ \left.\frac{c}{2}\sum_{j \in \mathcal{N}_i}\|\bB_{i|j}\bx_{i}+\bB_{j|i}\bx_{j}^{(k)} -\boldsymbol{b}_{i,j}\|_2^2 \right) \label{eq.xup} 
 \\
  \forall j \in \mathcal{N}_i: \, &\blambda_{i|j}^{(k+1)} = \blambda_{j|i}^{(k)} + c\big(\bB_{i|j}\bx_{i}+\bB_{j|i}\bx_{j}^{(k)} -\boldsymbol{b}_{i,j}\big). \label{eq.lup}
\end{align}

We can see that updating $\blambda_{i|j}^{(k+1)}$ requires only $\blambda_{j|i}^{(k)}, \bx_j^{(k)}$ and $\bx_i^{(k+1)}$, of which $ \blambda_{j|i}^{(k)}$ and $ \bx_j^{(k)}$ are already available at node $j$. Thus, at each iteration, node $i$ needs to broadcast only $\bx_{i}^{(k+1)}$ after which the neighboring nodes can update $\blambda_{i|j}^{(k+1)}$ themselves. As a consequence,  the dual variables do not need to be transmitted at all, except for the initialization step in which the initialized dual variables $\blambda^{(0)}$ should be transmitted. 

\subsection{Convergence of dual variables}
Consider two successive $\blambda$-updates in \eqref{eq.pdmmlUpdate}, in which we have
\begin{align}\label{eq.pdmmLamUp}
\blambda^{(k+2)} = \blambda^{(k)} + c(\bC\bx^{(k+2)} + 2\bPC\bx^{(k+1)} + \bC\bx^{(k)}),
\end{align}
as $\bP^2 = \bm I$. 
Let $\bar{H}_{p}= \operatorname{span}(\bC)+\operatorname{span}(\bPC)$ and $\bar{H}_{p}^{\perp} = \operatorname{null}(\bC^{\top}) \cap \operatorname{null}((\bPC)^{\top})$. Denote $\Pi_{\bar{H}_{p}}$ as the orthogonal projection onto $\bar{H}_{p}$. From \eqref{eq.pdmmLamUp}, we conclude that every two $\blambda$-updates affect only $\Pi_{\bar{H}_{p}} \blambda \in \bar{H}_{p}$, and $(\bI-\Pi_{\bar{H}_{p}} ) \blambda \in \bar{H}_{p}^{\perp}$
 remains the same.  Moreover, as shown in \cite{sherson2018derivation},
$\left(\bI-\Pi_{\bar{H}_{p}}\right)\blambda$ will only be permuted in each iteration and $\Pi_{\bar{H}_{p}}\blambda$ will eventually converge to $\blambda^*$ given by
\begin{equation}\label{eq.lamdaOptimum}
\blambda^* = -\left( \!\! \begin{array}{c} \bC^{\top} \\ (\bPC)^{\top}\end{array} \!\! \right)^{\!\!\dagger} \left( \!\! \begin{array}{c} \partial f(\bx^*) + c\bC^{\top}\!\bC\bx^* \\ \partial f(\bx^*) + c \bC^{\top}\!\bPC\bx^* \end{array} \!\! \right) + c\bC\bx^*.
\end{equation}
We thus can separate the dual variable into two parts:
\begin{align} 
\blambda^{(k)} &=\Pi_{\bar{H}_{p}}\blambda^{(k)}+(\bI-\Pi_{\bar{H}_{p}})\blambda^{(k)},\nonumber\\
&\rightarrow \blambda^* + \left\{ \begin{array}{ll} (\bI-\Pi_{\bar{H}_{p}}) \blambda^{(0)}, & k \text{ even}, \\ P\left(\bI-\Pi_{\bar{H}_{p}}\right) \blambda^{(0)}, & k \text{ odd}. \rule[4mm]{0mm}{0mm} \end{array} \right.
\label{eq.lamdaUpdata}
\end{align}
Below, $\bar{H}_{p}$ and $\bar{H}_{p}^{\perp}$ are respectively referred to as the convergent subspace and non-convergent subspace associated with PDMM, and similarly ${\Pi_{\bar{H}_{p}}} \blambda$ and ${\left(\bI-\Pi_{\bar{H}_{p}} \right)}\blambda$ are called the convergent and non-convergent   component of the dual variable, respectively. 

\section{Proposed approach using PDMM}\label{section:prop}
Having introduced the PDMM algorithm, we now introduce the proposed approach. To achieve a computationally and communicationally efficient solution for privacy preservation, one of the most used techniques is obfuscation by inserting noise, such as in the differential privacy technique. However, inserting noise usually compromises the function accuracy, because the updates are disturbed by noise. To alleviate this trade-off, we propose to insert noise in the non-convergent   subspace only so that the accuracy of the optimization solution is not affected (see also Remark \ref{rk:xInd}), thus achieving both privacy and accuracy at the same time. The proposed noise insertion method is referred to as subspace perturbation. Below, we first explain the proposed subspace perturbation in detail and then prove that it satisfies both the output correctness and individual privacy requirements stated in Section \ref{section:problemDef}. 
\subsection{Subspace perturbation}
Owing to the broadcasting property of the PDMM, after transmission of the initialized dual variables, the updated optimization variable $\bx_{i}^{k+1}$ is the only information transmitted in the network in each iteration. The main goal of privacy preservation thus becomes minimizing the information loss of $\bs_{i}$ by revealing $\bx_i^{(k+1)}$. From  \eqref{eq.xup}, $\bx_i^{(k+1)}$ is computed by \footnote{Note that $\bB_{i|j}^{\top}=\bB_{i|j}$.} 
\begin{align}\label{eq.xgradient}
  \bm 0 \in & \partial f_{i}(\bx_i^{(k+1)})+\sum_{j \in \mathcal{N}_i} \bB_{i|j}\blambda_{j|i}^{(k)} \nonumber \\&+c\sum_{j \in \mathcal{N}_i}(\bx_i^{(k+1)}-\bx_{j}^{(k)}  -\bB_{i|j}\boldsymbol{b}_{i,j})
  ,
\end{align}
as $\bB_{i|j}\bB_{j|i}=-\bI$ and $\bB_{i|j}\bB_{i|j}=\bI$. Note that only $\partial f_{i}(\bx_{i})$ contains information about the private data $\bs_i$. Given that at convergence $\bx_{j}^{(k)} \rightarrow \bx_{j}^*$ and $\Pi_{\bar{H}_{p}}\blambda^{(k)} \rightarrow \blambda^*$, we then propose to insert noise in the non-convergent   subspace $\bar{H}_{p}^{\perp}$, i.e., $(\bI-\Pi_{\bar{H}_{p}}) \blambda^{(0)}$, for perturbing $\partial f_{i}(\bx_{i})$, thereby protecting $\bs_i$. 
To ensure that $\bx_{i}^{(k+1)}$ does not reveal information about $\bs_i$, we require 
\begin{align} \label{eq.mutuPdmm}
 I(S_{i} ; X_{i}^{(k+1)})=0.
\end{align}
We have the following results. 
\begin{proposition}\label{prop:1}
Let $\{X_{i}\}_{i=1,\ldots,n}$ denote a set of continuous random variables with mean and variance 
$\mu_{X_i}$ and $\sigma^{2}_{X_i}$, respectively.
Let $\{Y_{i}\}_{i=1,\ldots,n}$ be a set of mutually independent random variables, i.e., $I(Y_{i}, Y_j) = 0, i\neq j$, which is also independent of $\{X_{i}\}_{i=1,\ldots,n}$, i.e., $I(X_{i}, Y_j) = 0$ for all $i,j \in \mathcal{N}$. 
Let $Z_{i}=X_{i}+Y_{i}$ and $W_{i}=X_{i}Y_{i}$, and let $Z_i' = Z_i/\sigma_{Z_i}$ and $W_i' = W_i/\sigma_{W_i}$ be the normalized random variables with unit variance. We have
\begin{align} 
 \lim_{\sigma^{2}_{Y_i}\rightarrow\infty} I(X_{1},\ldots,X_{n}; Z_{1},\ldots,Z_{n})=0, \label{eq.muZ}\\
 \lim_{\sigma^{2}_{Y_i}\rightarrow\infty} I(X_{1},\ldots,X_{n}; W_{1},\ldots,W_{n})=0. \label{eq.muW}
\end{align}
\end{proposition}
\begin{proof}
See Appendix~\ref{ap.1}.
\end{proof}
With Proposition \ref{prop:1}, we can see that the goal of privacy preservation \eqref{eq.mutuPdmm} can be realized by sufficiently obfuscating private data by using the subspace noise, i.e., the non-convergent   component of the dual variable $(\bI-\Pi_{\bar{H}_{p}}){\blambda}^{(0)}$.
We then conclude that a sufficient condition to ensure \eqref{eq.mutuPdmm} is given by
\begin{align}\label{eq.mutualInfo1}
\exists j \in \mathcal{N}_{i}^{h}: {\rm var}\left(((\bI-\Pi_{\bar{H}_{p}})\Lambda^{(0)})_{j|i}\right) \rightarrow \infty,
\end{align}
where $\mathcal{N}_{i}^{h}= \mathcal{N}_i \cap \mathcal{N}_{h}$ denotes the  honest neighbors of node $i$.
By inspecting the above condition, we have the following remarks.

\begin{remark}\label{rk:bound} (Information loss with finite variance)
In Proposition \ref{prop:1}, we proved that the information loss regarding private data becomes zero if the inserted noise has an infinitely large variance, which is impossible to realize in practice. To show that the proposed method is practically useful, i.e., when the noise variance is finite, the following result gives an  upper bound for the information leakage with Gaussian noise insertion.
\end{remark}
\begin{proposition}\label{prop:2}
Consider two independent random variables $X$ and $Y$ and let $Z=X+Y$, where $X$ denotes the private data and $Y$ denotes the inserted noise for protecting $X$. If we choose to insert Gaussian noise, the mutual information $I(X;Z)$ can be upper bounded by the following
\begin{align}
  I(X;Z)\leq \frac{1}{2}\log(1 +\sigma_{{X}}^2/\sigma_{Y}^2),
\end{align}
the equality holds when $X$ also has a Gaussian distribution.
\begin{proof}
See Appendix~\ref{ap.2}.
\end{proof}

\end{proposition}
From Proposition \ref{prop:2}, we can see that the if the variance of inserted Gaussian noise is 10 and 100 times the variance of the private data, the maximal information loss is only $7\times10^{-2}$ and $7\times10^{-3}$ bits, respectively. 
\begin{remark}\label{rk:null} ($\blambda^{(0)}\cap \bar{H}_{p}^\perp \neq \emptyset$ can be realized by randomly initializing $\blambda^{(0)}$). Of note, $[\bC~ \bPC]\in \mathbb{R}^{2M_m\times2N_n}$ can be viewed as a new graph incidence matrix with $2N_n$ nodes and $2M_m$ edges (see (c) in Fig. \ref{fig:graphTopo} for an example); we thus have $\mathrm{dim}(\bar{H}_{p}) \leq 2N_n-1$, and $\bar{H}_{p}^\perp$ is non-empty. For a connected graph with the number of edges not less than the number of nodes (i.e., $M_m\geq N_n$), we conclude that a randomly initialized $\blambda^{(0)} \in \mathbb{R}^{2M_m}$ will achieve $\blambda^{(0)}\cap \bar{H}_{p}^\perp \neq \emptyset$ with probability 1. 
\end{remark}

\begin{remark}\label{rk:xInd} (No trade-off between privacy and accuracy)
No matter how much noise is inserted in the non-convergent  subspace, the convergence of the optimization variable $\bx\rightarrow \bx^{*}$ is guaranteed. By inspecting \eqref{eq.pdmmSetup}, we can see that the $\bx$-update is independent of $(\bI-\Pi_{\bar{H}_{p}}) \blambda$ as $\blambda^{\top} \left(\bI-\Pi_{\bar{H}_{p}}\right)\bPC = 0$. 
\end{remark} 

Details of the proposed approach using PDMM are shown in algorithm \ref{alg:pdmm}. And the analysis of both output correctness and individual privacy is summarized below.
\subsection{Output correctness}
As proven in \cite{sherson2018derivation}, with strictly convex  $f(\bx_i)$, the optimization variable $\bx_{i}^{(k)}$ of each node $i$ of the PDMM is guaranteed to converge geometrically (linearly on a logarithmic scale) to the optimum solution $\bx_{i}^{*}$, regardless of the initialization of both $\bx^{(0)}$ and $\blambda^{(0)}$, thereby ensuring the correctness. Moreover, for convex functions that are not strictly convex, a slightly modified version called averaged PDMM (see Section \ref{subsec.lasso} for an example) can be used to guarantee the convergence. 
\subsection{Individual privacy}
As stated before, we consider both passive and eavesdropping adversaries in this paper. From \eqref{eq.mutualInfo1}, we conclude that the proposed algorithm achieves asymptotically perfect security against a passive adversary as long as the honest node has at least one honest neighbor, i.e., $\mathcal{N}_{i}^{h}\neq \emptyset$. Additionally, because privacy is ensured by $\blambda^{(0)}$, we conclude that the proposed approach is also secure against an eavesdropping adversary without requiring securely encrypted channels except for the first iteration for transmitting the initialized $\blambda^{(0)}$.

\begin{algorithm}[t]
\vskip -2pt
\caption{Privacy-preserving distributed optimization via subspace perturbation using PDMM}
\vspace{1.2mm}
\begin{algorithmic}[1] \label{alg:pdmm}
\STATE Each node $i \in \mathcal{N}$ initializes its optimization variable $\bx_{i}^{(0)}$ arbitrarily, and its dual variables $\blambda_{i|j}^{(0)},j \in \mathcal{N}_{i}$ are randomly initialized with a certain distribution with large variance (specified by the required privacy level). 
\STATE Node $i$ sends the initialized $\{\bx_{i}^{(0)},\blambda_{i|j}^{(0)}\}$ to its neighbor $j$ through secure channels \cite{dolev1993perfectly}.
\WHILE {$\|\bx^{(k)} - \bx^*\|_2 < $ threshold}
\STATE Activated a node uniformly at random, e.g., node $i$, updates $\bx_{i}^{(k+1)}$ using \eqref{eq.pdmmxUpdate}.
\STATE Node $i$ broadcasts $\bx_{i}^{(k+1)}$ to its neighbors $j\in \mathcal{N}_{i}$ (through non-secure channels).
\STATE \label{stepLamda} After receiving $\bx_{i}^{(k+1)}$, each neighbor $j\in \mathcal{N}_{i}$ updates $\blambda_{i|j}^{(k+1)}$ using \eqref{eq.pdmmlUpdate}.
\ENDWHILE
\end{algorithmic}
\vskip -2pt
\label{alg:ppdmm}
\end{algorithm}
\section{Proposed approach using other optimizers}\label{section:genFramework}
In this section, we demonstrate the general applicability of the proposed subspace perturbation method. In fact, the proposed method can be generally applied to any distributed optimizer if the introduced dual variables converge only in a subspace (i.e., there is a non-empty nullspace), which is indeed usually true, because these optimizers often work in a subspace determined by the incidence matrix of a graph. To substantiate this claim, we will show that the subspace perturbation also applies to ADMM and the dual ascent method. We then illustrate their differences by linking the convergent subspaces to their graph topologies.

\subsection{ADMM}
Given a standard distributed optimization problem stated in \eqref{eq.setup}, the augmented Lagrangian function of ADMM \cite{boyd2006randomized} is given by
\begin{equation}
L(\bx,\bv,\bz)= f(\bx)+\boldsymbol{v}^{\top}(\bm  M\bx+\bm W\boldsymbol{z})
+\frac{c}{2}\|\bm M\bx+\bm W\boldsymbol{z}-2\boldsymbol{d}\|^{2} ,
\end{equation}
where $\bm M \in \mathbb{R}^{2M_m \times N_n}$, like PDMM, is a matrix related to the graph incidence matrix, and $\bm M=[\bB_{+}^{\top},-\bB_{-}^{\top}]^{\top}$, $\bm W = \begin{bmatrix} -\bI^{\top} -\bI^{\top} \end{bmatrix}^{\top}\in \mathbb{R}^{2 M_m \times M_m}$. $\boldsymbol{\nu}\in \mathbb{R}^{2M_m}$ and $\boldsymbol{z}\in \mathbb{R}^{2M_m}$ are the introduced dual variables and auxiliary variable for constraints, respectively. The updating function of ADMM is given by 
\begin{align}
  \bx^{(k+1)}&=\arg \min _{\bx} L\left(\bx, \boldsymbol{z}^{(k)}, \boldsymbol{\nu}^{(k)}\right) \\ 
\boldsymbol{z}^{(k+1)}&=\arg \min _{\boldsymbol{z}} L\left(\bx^{(k+1)},\boldsymbol{z}, \boldsymbol{\nu}^{(k)}\right) \label{eq.admmzUpdata}\\ 
\boldsymbol{\nu}^{(k+1)}&=\boldsymbol{\nu}^{(k)}+c\left(\bm M \bx^{(k+1)}+\bm W\boldsymbol{z}^{(k+1)}-2\boldsymbol{d}\right) \label{eq.admmvUpdata}
\end{align}
By inspecting the $\boldsymbol{v}$-update in \eqref{eq.admmvUpdata}, we can see that it has a similar structure to that of the $\blambda$-update in \eqref{eq.pdmmLamUp} of PDMM.  Let $\bar{H}_{a}=\operatorname{span}(\bm M)+\operatorname{span}(\bm W)$. For feasibility, we assume $\boldsymbol{d} \in \bar{H}_{a}$. Similarly to PDMM, in this process, every $\boldsymbol{v}$-update has effects only in $\left(\Pi_{\bar{H}_{a}}\right) \boldsymbol{v} \in \bar{H}_{a}$ and leaves $(\bI-\Pi_{\bar{H}_{a}} ) \boldsymbol{v} \in \bar{H}_{a}^{\perp}, \bar{H}_{a}^{\perp} = \operatorname{null}(\bm M^{\top}) \cap \operatorname{null}((\bm W)^{\top})$ unchanged. 

Of note, ADMM is not a broadcasting protocol, i.e., it requires pairwise communications for transmitting the dual variable and auxiliary variable. The effects on privacy are addressed below.
\begin{remark}
(Revealing the dual variables $\bm v$ and the auxiliary variable $\bm z$ will not disclose more information than revealing the optimization variable $\bx$.) By inspecting \eqref{eq.admmzUpdata} and \eqref{eq.admmvUpdata}, we can see that at each iteration $S_{i}\rightarrow X_{i}^{(k)} \rightarrow Z_{i|j}^{(k)}$ and $S_{i} \rightarrow X_{i}^{(k)} \rightarrow V_{i|j}^{(k)}$ both form a Markov chain. As a consequence, for each honest node $i$, we have 
\begin{align}
 \forall j\in \mathcal{N}_i: ~ I(S_{i} ; Z_{i|j}^{(k+1)}) \leq I(S_{i} ; X_{i}^{(k+1)}), \\
 \forall j\in \mathcal{N}_i: ~ I(S_{i} ; V_{i|j}^{(k+1)}) \leq I(S_{i} ; X_{i}^{(k+1)}),
\end{align}
by the data processing inequality \cite{cover2012elements}.
\end{remark}
With the above result, we conclude that the proof for output correctness and individual privacy using ADMM follows a similar structure as that of PDMM; the only difference is the constructed subspace. 
We thus conclude that privacy-preserving ADMM can be achieved by subspace perturbation, i.e., randomly initializing its dual variables $\boldsymbol{\nu}^{(0)}$ with a certain distribution with a large variance.

\subsection{Dual ascent method}
The Lagrangian of the dual ascent method for solving \eqref{eq.setup} is given by
\begin{equation}
L(\bx,\bm u)=f(\bx)+\boldsymbol{u}^{\top}(\bB\bx-\boldsymbol{b}),
\end{equation}
where $\boldsymbol{u}\in \mathbb{R}^{M_m}$ is the introduced dual variable. The updating function is given by 
\begin{align}
\bx^{(k+1)}=\arg \min _{\bx} L\left(\bx, \boldsymbol{u}^{(k)}\right) \\ 
\boldsymbol{u}^{(k+1)}=\boldsymbol{u}^{(k)}+t^{(k)}
\left(\bB\bx^{(k+1)}-\boldsymbol{b} \right), \label{eq.dauUpdate}
\end{align}
where $t^{(k)}$ denotes the step size at iteration $k$. 
Likewise, the $\boldsymbol{u}$-update in \eqref{eq.dauUpdate} has a similar structure as the $\blambda$-update of PDMM, wherein the convergent subspace is $\bar{H}_{d}= \operatorname{span}(\bB)$. Moreover, as with ADMM, revealing the dual variable $\bm u$ will not increase the information loss beyond revealing the optimization variable $\bx$. Hence, we conclude that the proposed subspace perturbation method also works for the dual ascent method.

\subsection{Linking graph topologies and subspaces}\label{subsec:topo}
Thus far, we have shown that the dual variable updates of PDMM, ADMM and the dual ascent method are dependent only on their corresponding subspaces: $\bar{H}_{p}=\operatorname{span}(\bC)+\operatorname{span}(\bPC)$, $\bar{H}_{a}=\operatorname{span}(\bm M)+\operatorname{span}(\bm W)$ and $\bar{H}_{d}= \operatorname{span}(\bB)$. Note that each of the matrices $[\bC ~\bPC]$, $[\bm M ~ \bm W]$ and $\bB$ can be seen as an incidence matrix of an extended graph, therefore they all have a non-empty left nullspac for subspace perturbation as long as $m\geq n$ (Remark \ref{rk:null}). To examine the appearance of these constructed graphs, in Fig \ref{fig:graphTopo} we give an example of these extended graphs and provide illustrative insights into the differences among these optimizers. 

\begin{figure*}
    \centering
    \includegraphics[width=0.8\linewidth]{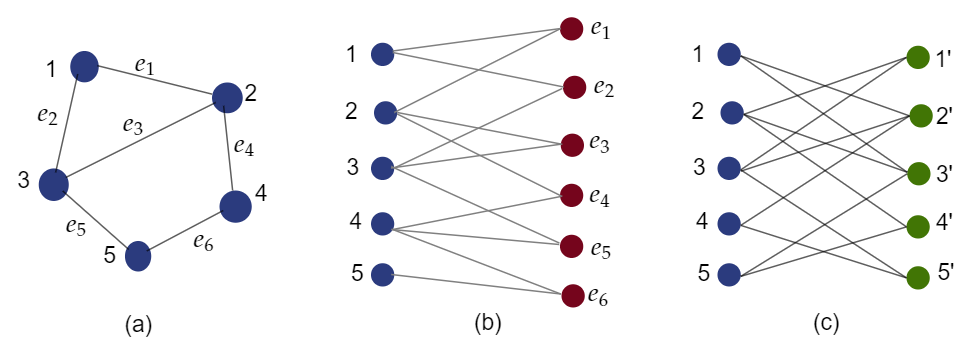}
\caption{An example of graph topologies associated with dual ascent, ADMM and PDMM with $u=1$: (a) A graph with $n=5$ nodes and $m=6$ edges. (b) The bipartite graph constructed by ADMM with $n+m$ nodes and $2m$ edges. (c) The bipartite graph constructed by PDMM with $2n$ nodes and $2m$ edges. }
\label{fig:graphTopo}
\end{figure*}

\section{Applications}\label{section:applications}
To demonstrate the potential of the proposed approach to be used in a wide range of applications, we now present the application of the proposed subspace perturbation to three fundamental problems: distributed average consensus, distributed least squares and distributed LASSO, because they serve as building blocks for many other signal processing tasks, such as denoising~\cite{pang2015optimal}, interpolation~\cite{narang2013signal}, machine learning \cite{tibshirani1996regression,wright2008robust} and compressed sensing \cite{candes2006near,donoho2006compressed}. 
We first introduce the application and then perform an analysis of the individual privacy.
We will continue using PDMM to introduce the details, but the numerical results of using all the discussed optimizers will be presented in the next section. 
\subsection{Privacy-preserving distributed average consensus}\label{subsec:consensus}
Privacy-preserving distributed average consensus is to estimate the average of all the nodes' initial state values over a network and keep each node's initial state value unknown to others. 
Such privacy-preserving solutions are highly desired in practice. For example, in face recognition applications, computing mean faces is usually required, thus prompting privacy concerns. Here, a group of people may cooperate to compute the mean face, but none of them would like to reveal their own facial images during the computation. 

The optimization problem setup \eqref{eq.setup} becomes 
\begin{equation} \label{eq.setupAve}
\begin{array}{ll}{\displaystyle \min_{\bx_i}} & {{\displaystyle\sum_{i \in \mathcal{N}} \frac{1}{2}\|\bx_i - \bs_i\|^2_2}} \\ {\text { s.t. }} & {\bx_i=\bx_j, {\forall} (i,j)\in \mathcal{E}}.\end{array}
\end{equation}
The optimum solution for each optimization variable is $\bx^{*}_i=n^{-1}\sum_{i \in \mathcal{N}} \bs_i $ and 
$\bx^* = (\bx^{*}_1,\ldots,\bx^{*}_n)^{\top}$.
With PDMM, the $\bx_i^{(k+1)}$ computed by \eqref{eq.xgradient} becomes 
\begin{align}
&\bx_i^{(k+1)} = \frac{\bs_{i} + \sum_{j\in \mathcal{N}_{i}} \left(c\bx_j^{(k)} - \bB_{i|j}\blambda^{(k)}_{j|i} \right)}{1+cd_i}. \label{eq.xup_ave}
\end{align}
\subsubsection{Individual privacy}
Based on $I(S_{i};X_{i}^{(k+1)})$, we analyze the individual privacy. 
With \eqref{eq.lamdaUpdata}, the numerator of \eqref{eq.xup_ave} can be separated as 
\begin{align}
&\bs_{i} + \sum_{j\in \mathcal{N}_{i}} c\bx_j^{(k)} - \sum_{j\in \mathcal{N}_{i}^{c}} \bB_{i|j}\blambda^{(k)}_{j|i}- \sum_{j\in \mathcal{N}_{i}^{h}} \bB_{i|j} \left(\Pi_{\bar{H}_{p}}\blambda^{(k)}\right)_{j|i}\nonumber \\
& - \sum_{j\in \mathcal{N}_{i}^{h}} \bB_{i|j} \left( P^{k} \left(\bI - \Pi_{\bar{H}_{p}}\right) \blambda^{(0)}\right)_{j|i},
\label{eq.xl}
\end{align}
where $\mathcal{N}_{i}^{c}=\mathcal{N}_i \cap \mathcal{N}_{c}$ denotes the corrupted neighbourhood of node $i$.
At convergence, we have $\bx_{i}^{(k)}\rightarrow \bx_{i}^* $ and $\Pi_{\bar{H}_{p}}\blambda^{(k)} \rightarrow \blambda^*$ given by \eqref{eq.lamdaOptimum}. We can see that the last term in \eqref{eq.xl} is unknown to the adversary, and it can protect the private data $\bs_{i}$ with a sufficiently large perturbation. 
By applying \eqref{eq.muZ} in Proposition \ref{prop:1} with dimension $n=1$,
we can achieve $I(S_{i};X_{i}^{(k)}) = 0$ under the condition of \eqref{eq.mutualInfo1}.
As for the lower bound of information leakage, the revealed information would be the partial sums of all the honest components (connected subgraphs consist of honest nodes only) after removal of all the corrupted nodes \cite{JaneICASSP2020}. Denote the node set of the honest component that the honest node $i$ belongs to as $\mathcal{C}$. The lower bound of information leakage for node $i$ thus becomes $I_{i}^{*}=I(S_{i};\sum_{j \in \mathcal{C}}S_{j})$. 
Additionally, the cumulative information leakage does not exceed $I_{i}^{*}$, because manipulating the information over all the iterations will reveal only the partial sums. 
Hence, the proposed algorithm obtains asymptotically perfect security in distributed average consensus applications, because \eqref{eq.perPriv} is satisfied. 

\subsection{Privacy-preserving distributed least squares}
Privacy-preserving distributed least squares aims to find a solution for a linear system (here we consider an overdetermined system in which there are more equations than unknowns) over a network in which each node has only partial knowledge of the system and is only able to communicate with its neighbors, and in the meantime the local information held by each node should not be revealed to others. More specifically, here the local information of node $i$ means both the input observations, denoted by $\bQ_i \in \mathbb{R}^{p_{i} \times u}, \, p_i>u$ and decision vector, denoted by $\by_i\in \mathbb{R}^{p_{i}}$. That is, each node $i$ has $p_i$ observations, and each contains an $u$-dimensional feature vector. Collecting all the local information, we thus have $\bQ=[\bQ^{\top}_1,\dots,\bQ^{\top}_n]^{\top} \in \mathbb{R}^{P_n \times u}$ and $\by=[\by^{\top}_1,\dots,\by^{\top}_n]^{\top}\in \mathbb{R}^{P_n}$, recall $P_n=\sum_{i\in \mathcal{N}} p_{i}$.  The reason for identifying the private data of each node $i$ as the local information $\boldsymbol{y}_i$ and $\bQ_i$ is that they usually contain users' sensitive information. Take the distributed linear regression as an example, which is widely used in the field of machine learning, and consider the case that several hospitals want to collaboratively learn a predictive model by exploring all the data in their datasets. Although such collaborations are limited because they must comply with policies such as the general data protection regulation (GDPR) and  because individual patients/customers may feel uncomfortable with revealing their private information to others, such as insurance data and health condition. 

The least-squares problem in a distributed network can be formulated as a distributed linearly constrained convex optimization problem, and the problem setup in \eqref{eq.setup} becomes
\begin{equation} \label{eq.LS}
\begin{array}{ll}{\displaystyle \min_{\bx_i}} & {\sum_{i \in \mathcal{N}} \frac{1}{2}\|\by_{i}-\bQ_i\bx_{i}\|^{2}_{2}} \\ {\text { s.t. }} & {\bx_i=\bx_j, {\forall} (i,j)\in \mathcal{E}},\end{array}
\end{equation}
where the optimum solution is given by $\bx_i^{*}=(\bQ^{\top}\bQ)^{-1}\bQ^{\top}\by \in \mathbb{R}^{u}, \forall i\in \mathcal{N}$.
With PDMM, the $\bx$-update \eqref{eq.pdmmxUpdate} for node $i$ becomes
\begin{align}\label{eq.xUpdate_LS}
&\bx_{i}^{(k+1)}=\\
&(\bQ_i^{\top}\bQ_i+c d_i \bI)^{-1} 
\left( \bQ_i^{\top}\by_i \nonumber + \sum_{j\in \mathcal{N}_{i}} (c\bx_j^{(k)} - \bB_{i|j}\blambda^{(k)}_{j|i}) \right). 
\end{align}

\subsubsection{Individual privacy}
To analyze how much information the adversary would learn about the private data $\bQ_i$ and $y_{i}$ by observing $\bx_{i}^{(k+1)}$, with \eqref{eq.lamdaUpdata}, we can separate \eqref{eq.xUpdate_LS} into 
\begin{align}\label{eq.xmu}
  &(\bQ_i^{\top}\bQ_i+c d_i \bI)^{-1} \left( \sum_{j\in \mathcal{N}_i^{h}} \big(c\bx_j^{(k)}-\bB_{i|j}\big( P^{k} \Pi_{\bar{H}_{p}}\blambda^{(k)}\big)_{j|i}\big)\right. \nonumber\\ 
  &- \left. \sum_{j\in \mathcal{N}_i^{h}}\bB_{i|j}\big( P^{k} (\bI - \Pi_{\bar{H}_{p}})\blambda^{(0)}\big)_{j|i} +\bQ_i^{\top}\by_i+ c_{p} \right),
\end{align}
where $c_{p}=\sum_{j\in \mathcal{N}_{i}^{c}} \big(c\bx_{j}^{(k)}-\bB_{i|j}\blambda_{j|i}^{(k)}\big)$
is known by the adversary. Because both $\bx_{j}^{(k)} \rightarrow \bx_{j}^*$ and $\Pi_{\bar{H}_{p}}\blambda^{(k)} \rightarrow \blambda^*$, again, we use subspace noise, i.e., the non-convergent   component of the dual variable $(\bI - \Pi_{\bar{H}_{p}})\blambda^{(0)}$, to protect the private data. More specifically, the information loss regarding the private data $\bQ_i$ and $\by_i$ can be quantified by the mutual information $I(Q_{i},X_{i}^{(k+1)})$ and $I(Y_{i},X_{i}^{(k+1)})$, respectively.
Given that $\blambda^{(0)}$ is independent of both $\bQ_{i}$ and ${\by}_{i}$, by applying Proposition \ref{prop:1}, we again can achieve $I(Q_{i},X_{i}^{(k+1)})=0$ and $I(Y_{i},X_{i}^{(k+1)})=0$ as long as
\eqref{eq.mutualInfo1} is satisfied. 
Additionally, collecting all the information over the entire iterative process will not increase the information loss. 
We therefore conclude that the sufficient condition \eqref{eq.perPriv} is satisfied.
Hence, the proposed approach achieves asymptotically perfect security in the distributed least squares problem as well. 
\begin{figure*}[ht]
\begin{subfigure}{0.32\textwidth}
\includegraphics[width=0.9\linewidth]{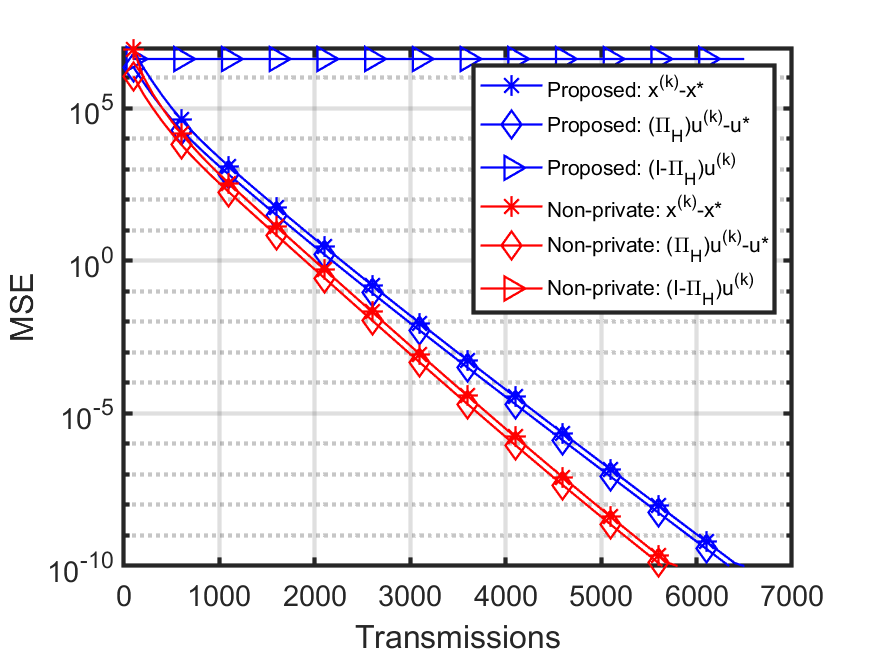} 
\caption{}
\end{subfigure}
\begin{subfigure}{0.32\textwidth}
\includegraphics[width=0.9\linewidth]{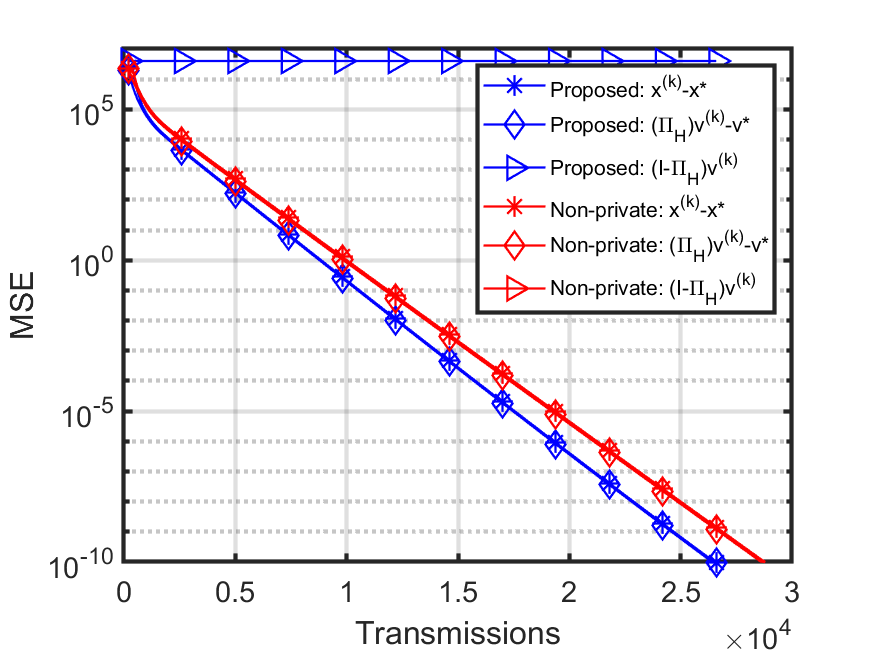}
\caption{}
\end{subfigure}
\begin{subfigure}{0.32\textwidth}
\includegraphics[width=0.9\linewidth]{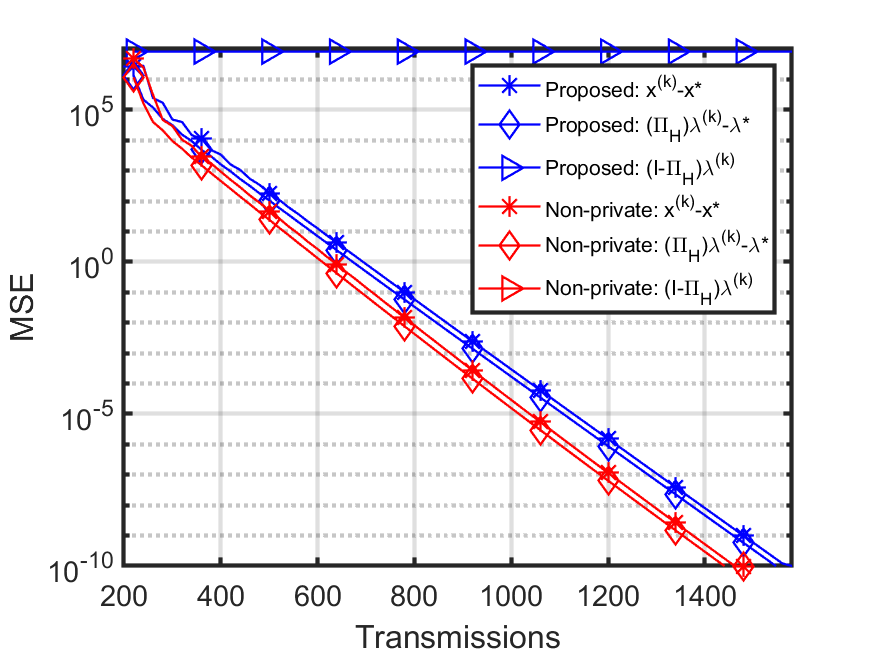}
\caption{}
\end{subfigure}
\caption{Distributed average consensus with two different initializations of the dual variable with a variance of $10^{6}$: convergence of the optimization variable, the convergent and non-convergent   component of the dual variable, using (a) dual ascent, (b) ADMM and (c) PDMM.}
\label{fig.aveCon}
\end{figure*}

\begin{figure*}[ht]
\begin{subfigure}{0.32\textwidth}
\includegraphics[width=0.9\linewidth]{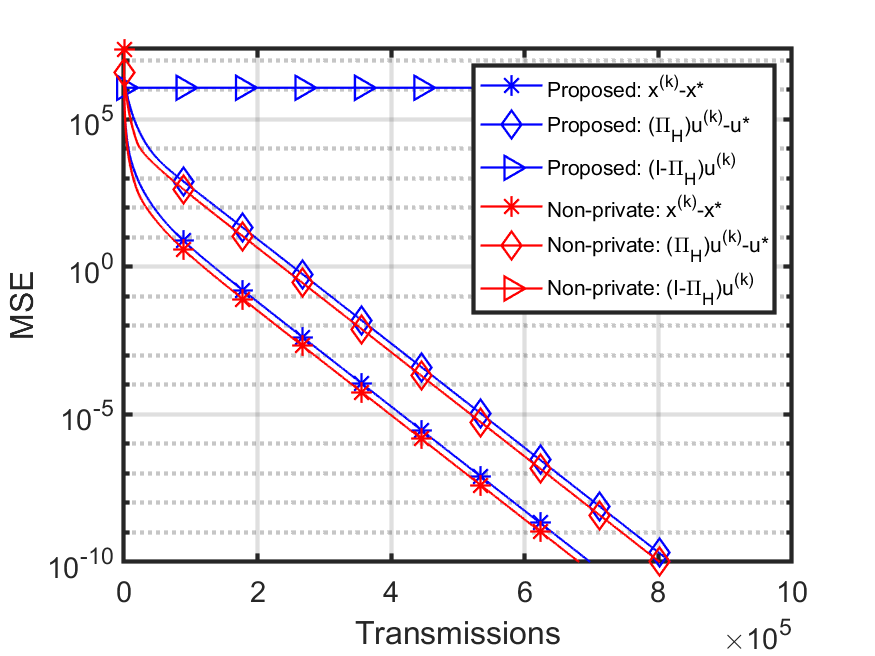} 
\caption{}
\end{subfigure}
\begin{subfigure}{0.32\textwidth}
\includegraphics[width=0.9\linewidth]{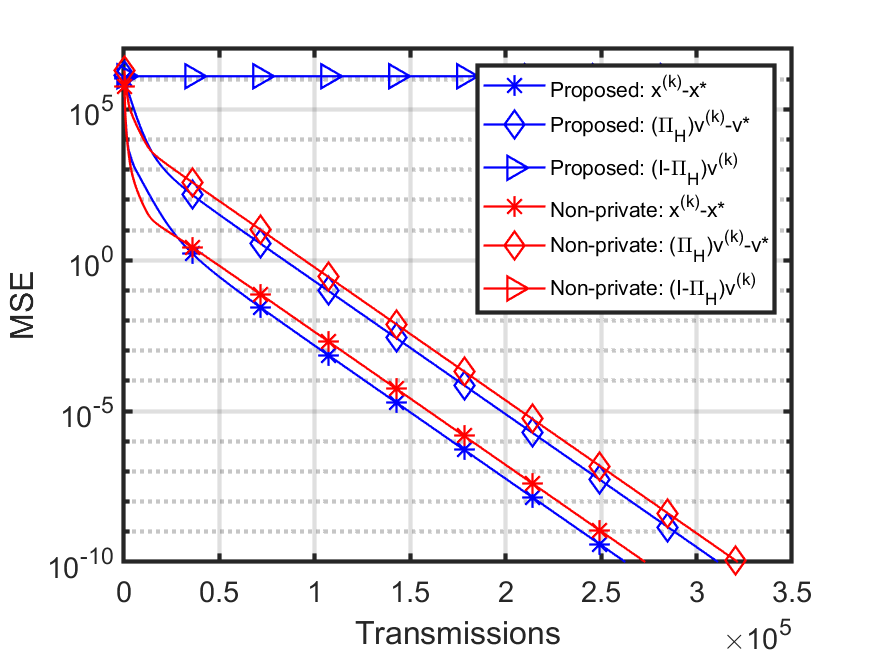}
\caption{}
\end{subfigure}
\begin{subfigure}{0.32\textwidth}
\includegraphics[width=0.9\linewidth]{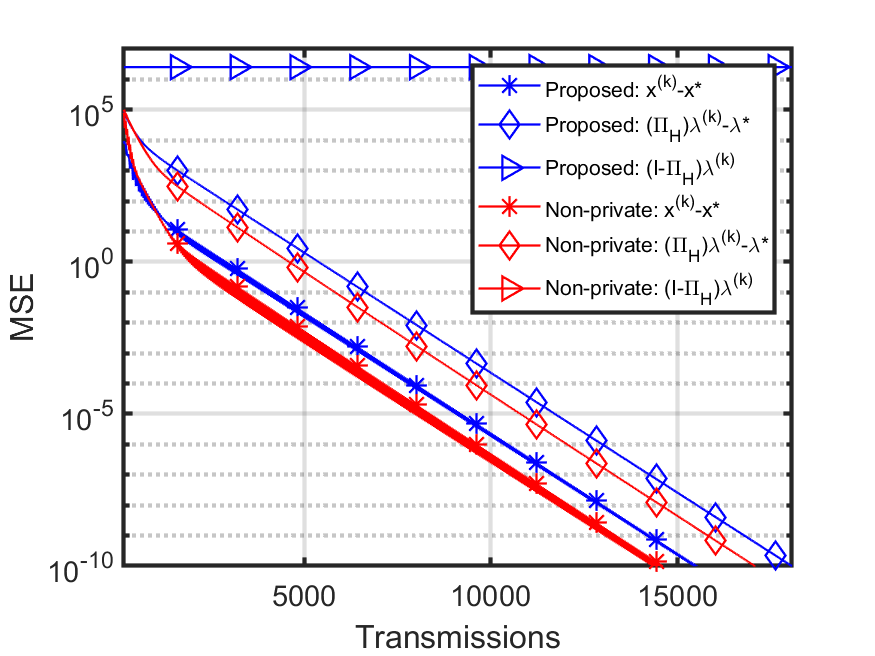}
\caption{}
\end{subfigure}
\caption{Distributed least squares with two different initializations of the dual variable with a variance of $10^{6}$: convergence of the optimization variable, the convergent and non-convergent   component of the dual variable of (a) dual ascent, (b) ADMM and (c) PDMM.}
\label{fig.lsCon}
\end{figure*}

\begin{figure*}[ht]
\begin{subfigure}{0.32\textwidth}
\includegraphics[width=0.9\linewidth]{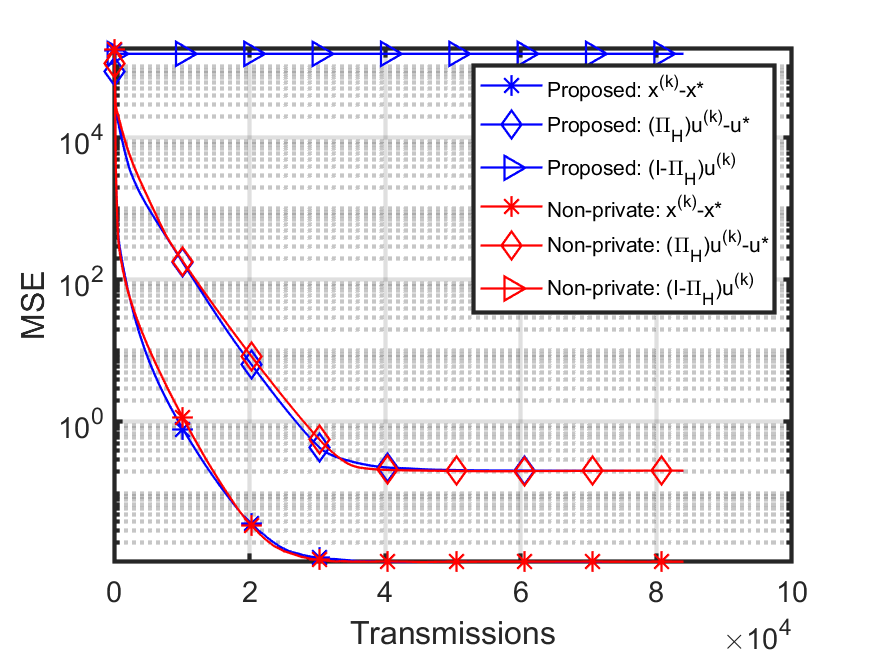} 
\caption{}
\end{subfigure}
\begin{subfigure}{0.32\textwidth}
\includegraphics[width=0.9\linewidth]{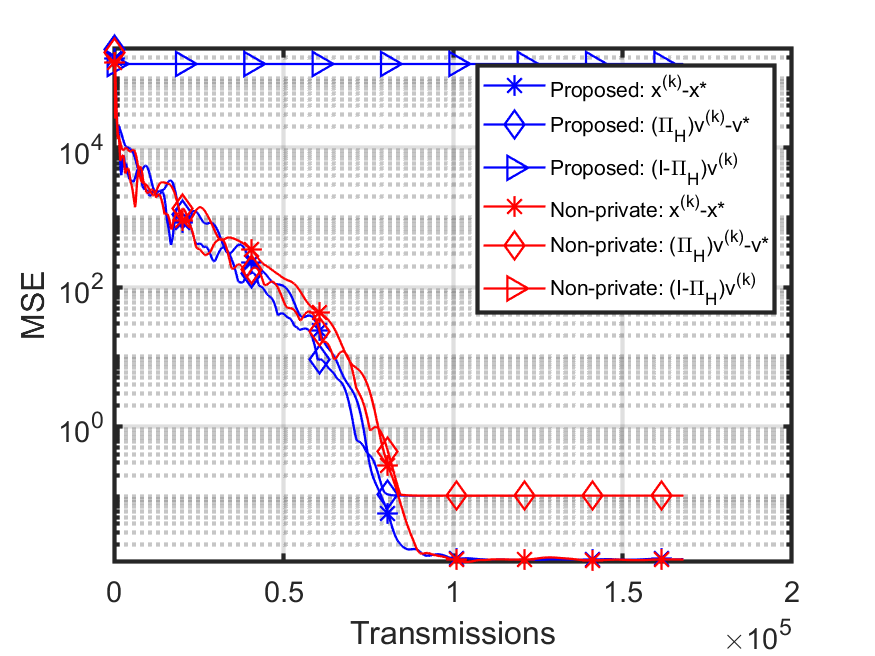}
\caption{}
\end{subfigure}
\begin{subfigure}{0.32\textwidth}
\includegraphics[width=0.9\linewidth]{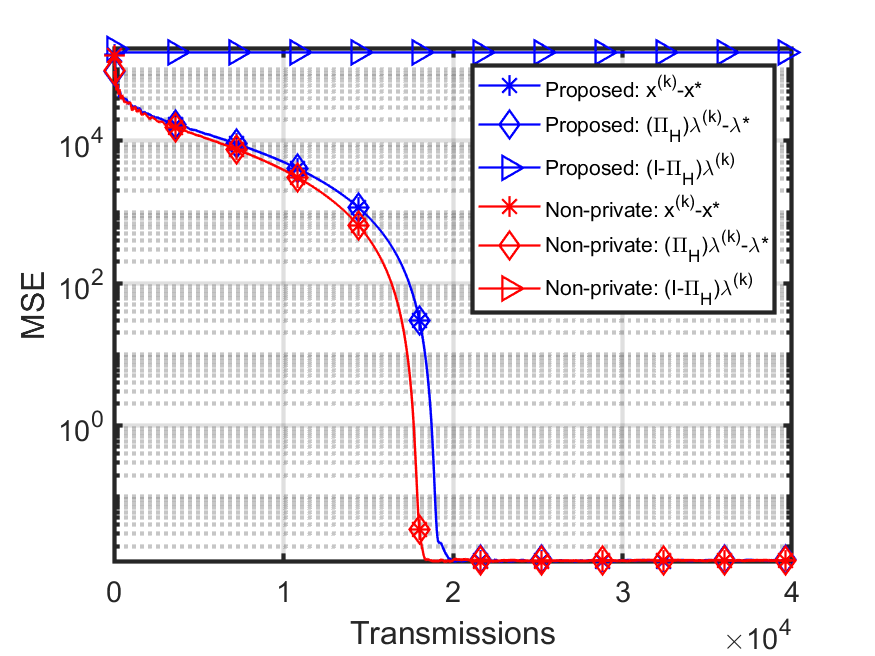}
\caption{}
\end{subfigure}
\caption{Distributed LASSO with two different initializations of the dual variable with a variance of $10^{6}$: convergence of the optimization variable, the convergent and non-convergent   component of the dual variable of (a) dual ascent, (b) ADMM and (c) PDMM.}
\label{fig.lassoCon}
\end{figure*}

\subsection{Privacy-preserving distributed LASSO}\label{subsec.lasso}
Privacy-preserving distributed LASSO aims to securely find a sparse solution when solving an underdetermined system (in which the number of equations is less than number of unknowns). We thus have a network similar to the previous least squares section but with the dimension $P_n< u$ to ensure an underdetermined system. The distributed LASSO is formulated as a $\ell_1$-regularized distributed least squares problem given by
\begin{align} \label{eq.Lasso}
{ \displaystyle \min_{\bx_i}}&  ~{\sum_{i \in \mathcal{N}} \left(\frac{1}{2}\|\by_{i}-\bQ_i\bx_{i}\|^{2}_{2}+\alpha|\bx_{i}|\right)}  \nonumber\\ {\text { s.t. }} &
{\bx_i=\bx_j, {\forall} (i,j)\in \mathcal{E}}
\end{align}
where $\alpha$ the constant controlling the sparsity of the solution. Because the objective function is convex but not strictly convex, we use averaged PDMM to ensure convergence, the $\bx$-updating function remains the same, and the $\blambda$-updating function in \eqref{eq.pdmmlUpdate} is replaced with a weighted average by 
\begin{align}\label{eq.avepdmmlUpdata}
  &\blambda^{(k+1)} = \theta(\blambda^{(k)}+c\bC(\bx^{(k+1)} -\bx^{(k)}))\nonumber\\ &+ (1-\theta)\left(\bP\blambda^{(k)} + c(\bC\bx^{(k+1)} + \bPC\bx^{(k)}-2\boldsymbol{d})\right),
\end{align}
where $0<\theta<1$ is the constant controlling the average weight. The output correctness is ensured by simply replacing the equation \eqref{eq.pdmmlUpdate} in step \ref{stepLamda} of algorithm \ref{alg:pdmm} with the above equation. The analysis proof of individual privacy follows similarly as the example for distributed least squares described above. Hence, with \eqref{eq.mutualInfo1}, we are able to achieve asymptotically perfect security in distributed LASSO. 

\begin{figure*}[ht]
\begin{subfigure}{0.32\textwidth}
\includegraphics[width=0.9\linewidth]{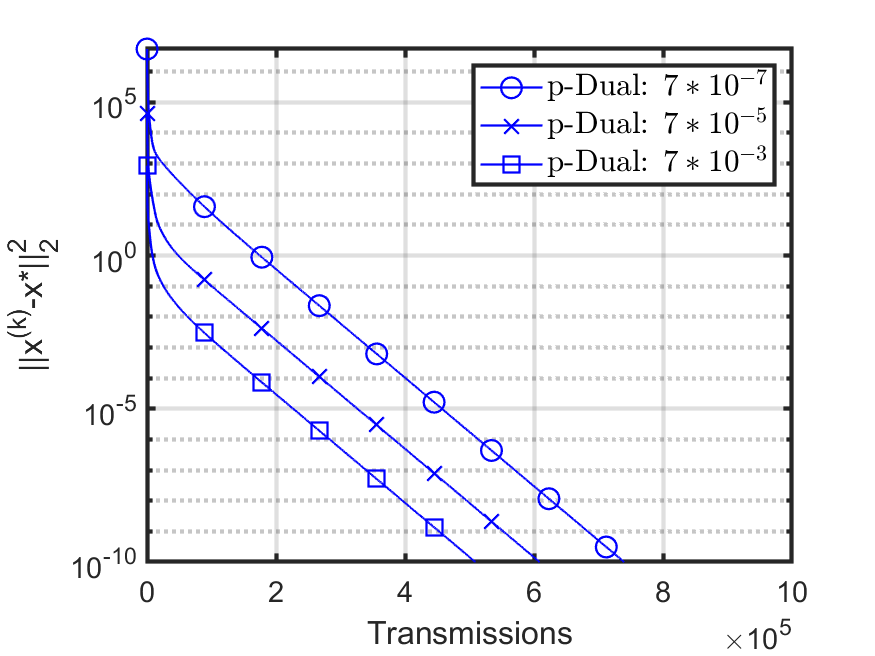} 
\caption{}
\end{subfigure}
\begin{subfigure}{0.32\textwidth}
\includegraphics[width=0.9\linewidth]{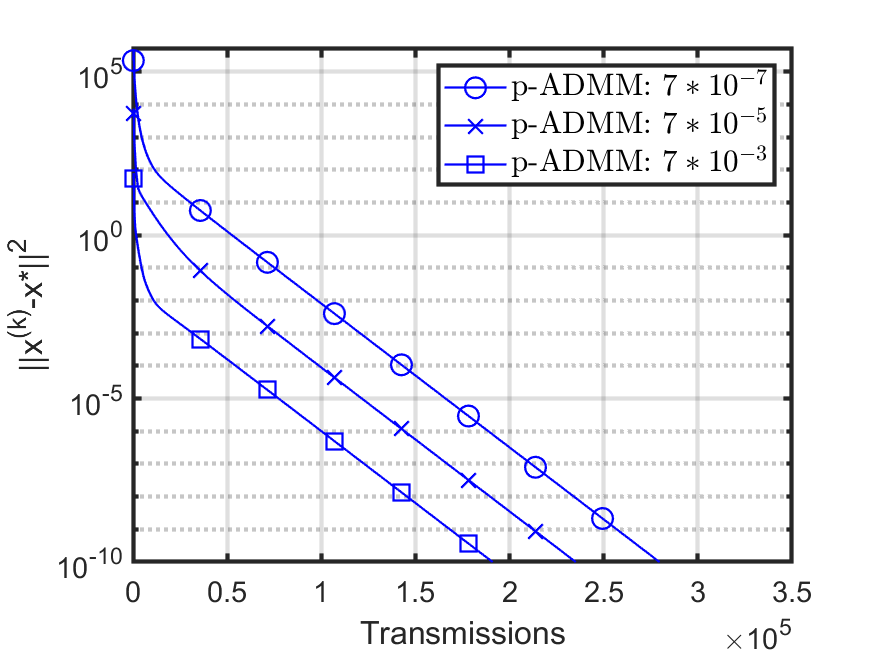}
\caption{}
\end{subfigure}
\begin{subfigure}{0.32\textwidth}
\includegraphics[width=0.9\linewidth]{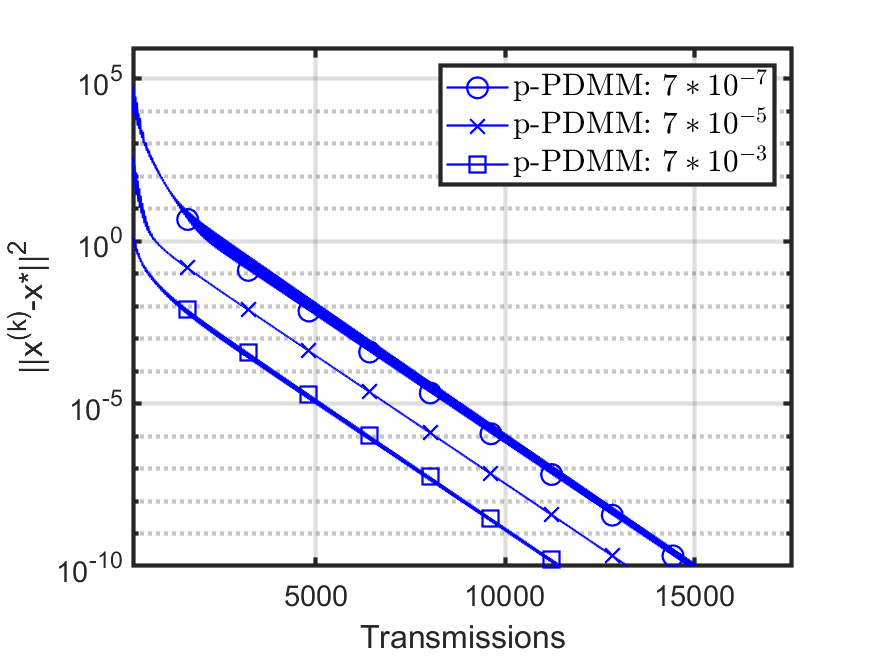}
\caption{}
\end{subfigure}
\caption{Convergence of the optimization variable in terms of three different privacy levels, i.e., approximately $7\times10^{-3}$, $7\times10^{-5}$ and $7\times10^{-7}$ bits, for (a) proposed dual ascent (p-Dual), (b) proposed ADMM (p-ADMM) and (c) proposed PDMM (p-PDMM) in a distributed least squares application.}
\label{fig.lsPri}
\end{figure*}

\section{Numerical results} \label{section:numerical}
In this section, several numerical tests\footnote{The code for reproducing these results is available at https://github.com/qiongxiu/PrivacyOptimizationSubspace} are conducted to demonstrate both the generally applicability and the benefits of the proposed subspace perturbation in terms of several important parameters including accuracy, convergence rate, communication cost and privacy level. 

We simulated a distributed network by generating a random graph with $n=20$ nodes and the communication radius $r$ was set at $r^2 = 2\frac{\log n}{n}$ so ensure that the graph is connected with high probability \cite{dall2002random}. 
For simplicity, all the private data in each application are randomly generated from a Gaussian distribution with unit variance, and all the optimization variables are initialized with zeros. Additionally, from Proposition \ref{prop:2}, we initialize the dual variables with a Gaussian distribution with a variance of $10^{2}, 10^{4}$, and $10^{6}$, thereby ensuring that the privacy loss is approximately $7\times10^{-3}$, $7\times10^{-5}$, and $7\times10^{-7}$ bits, respectively. 

\subsection{General applicability}
In Fig. \ref{fig.aveCon},~\ref{fig.lsCon} and \ref{fig.lassoCon}, we compare the convergence behavior of the proposed subspace perturbation methods (blue lines) with traditional non-private approaches (red lines) by using three distributed optimizers in three applications: distributed average consensus, least squares and LASSO. More specifically, the blue lines indicate that the dual variables are randomly initialized with a variance of $10^{6}$, such that the non-convergent component (blue line with triangle markers) can protect the private data, whereas the red lines mean that the dual variables are initialized within the convergent subspace, and the private data are therefore not protected, because the non-convergent   component is zero (the red lines with triangle markers are not shown in the plots).
We can see that the proposed approach has no effect on the accuracy, because all the optimization variables converge to the same optimum solution as the non-private counterparts. Overall, we can conclude that 
\begin{enumerate}
  \item the proposed approach is able to achieve both accuracy and privacy simultaneously;
  \item it is able to solve a variety of convex problems; 
  \item it is generally applicable to a broad range of distributed convex optimization methods.
\end{enumerate}
\subsection{Privacy level-invariant convergence rate}
Another important aspect to quantify the performance is the influence on the convergence rate when considering privacy. Because the convergence rates of the discussed distributed optimizers depend only on the underlying graph topologies rather than the initializations, initializing the dual variable with greater variance, i.e., a higher privacy level can be obtained, will therefore not change the convergence rate; and it will only result in only a larger offset in the initial error. To validate these results, in Fig. \ref{fig.lsPri} we show the convergence behavior of the proposed approaches in the distributed least squares problem under three different privacy levels by assigning the noise variance as $10^{2}, 10^{4}$ or $10^{6}$.
As expected, in all optimizers, the convergence rate remains identical regardless of the privacy level. 
\begin{figure}
  \centering
  \includegraphics[width=.40\textwidth]{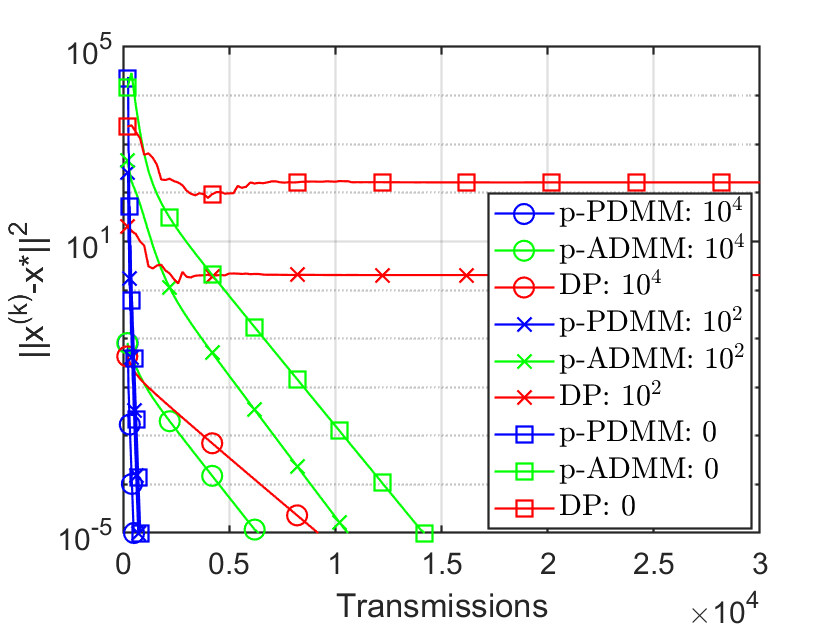}
  \caption{Performance comparison: convergence behavior under three differential privacy levels using the proposed p-PDMM, p-ADMM and an existing differential privacy (DP) approach.}
  \label{fig:dpComp}
\end{figure}
\subsection{Comparison with differential privacy}
To demonstrate the benefits of the proposed method, in Fig. \ref{fig:dpComp}, we compare the proposed approaches (both p-PDMM and p-ADMM) with a differential privacy approach \cite{nozari2017differentially} by using the distributed average consensus application. We consider three cases in which the noise variance is set at $0, 10^{2}, 10^{4}$. We can see that the accuracy of the differential privacy approach decreases with increasing privacy level. Hence, there is a trade-off between privacy and accuracy. Additionally, the convergence rates of differential privacy approaches will also be affected when increasing the level of privacy, because noise is inserted at every iteration, and a higher privacy level will also result in a slower convergence rate.

\begin{figure}
  \centering
  \includegraphics[width=.40\textwidth]{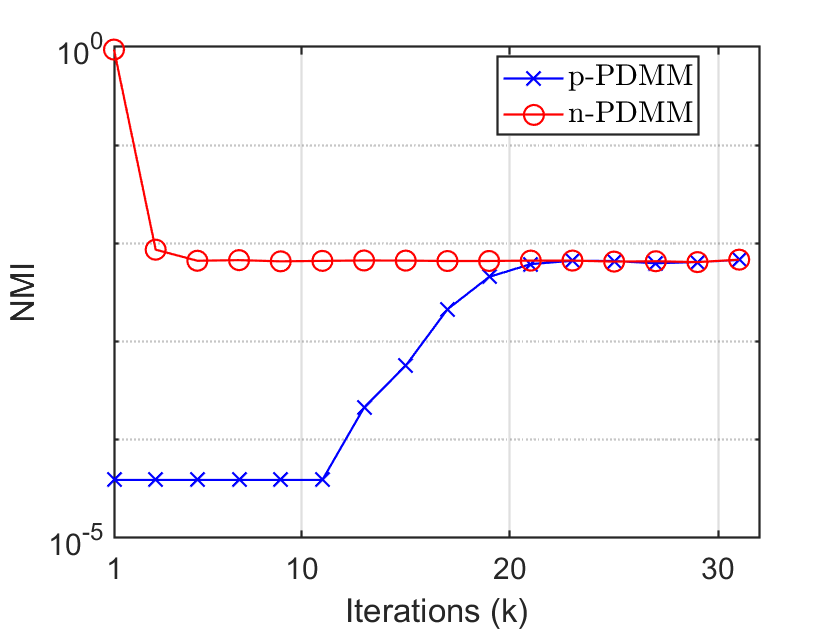}
  \caption{Normalized mutual information of an arbitrary node $i$ (i.e., 
  $\frac{I(S_{i}; X_{i}^{(k+1)})}{I(S_{i}; S_{i})}$) using the proposed p-PDMM and non-private PDMM (n-PDMM) for each iteration.}
  \label{fig:nmi}
\end{figure}
\subsection{Information loss over the iterative process}
To visualize how the information loss behaves during the iterative process, we perform $10^{5}$ Monte Carlo simulations and estimate the normalized mutual information $\frac{I(S_{i}; X_{i}^{(k+1)})}{I(S_{i}; S_{i})}$ of distributed average consensus using the non-parametric entropy estimation toolbox (npeet) \cite{ver2000non}. In Fig. \ref{fig:nmi}, we show the estimated normalized mutual information of the proposed p-PDMM with a noise variance of $10^{6}$ and traditional non-private PDMM, in which the dual variables are initialized with zeros. We can see that both approaches ultimately have the same information loss; i.e., the lower bound is $I^{*}_{i}=I(S_{i};X_{i}^{*})$, because they all converge to the same average result $\bx^{*}_i=n^{-1}\sum_{i \in \mathcal{N}} \bs_i $. As expected, the information loss of the proposed p-PDMM never exceeds the lower bound; hence, the proposed approach achieves perfect security. However, the n-PDMM reveals all the information about $s_i$ at the first iteration, i.e.,  $I(S_{i}; X_{i}^{(1)})=I(S_{i}; S_{i})$ because $\bx_i^{(1)} = \frac{\bs_{i}}{1+cd_i}$ based on \eqref{eq.xup_ave}; thus the privacy is not protected at all. 

\section{Conclusions}\label{section:conclusion}
In this paper, a novel and general subspace perturbation method was proposed for privacy-preserving distributed optimization. As a noise insertion approach, this method is more practical than SMPC based approaches in terms of both computation and communication costs. Additionally, by inserting noise in subspace, it circumvents the trade-off between privacy and accuracy in traditional noise insertion approaches such as differential privacy. Moreover, the proposed method is generally applicable to various optimizers and all convex problems. 
Furthermore,  we consider both eavesdropping and passive adversary models; in the former case, only secure channel encryption in the initialization is required, and in the latter case, the private data of each honest node are protected as long as the node has one honest neighbor. 
\appendices
\section{Proof of Proposition~\ref{prop:1}}
\label{ap.1}
\begin{proof}
\begin{align*}
 I(X_{1},&\ldots,X_{n}; Z_{1},\ldots,Z_{n})\\
 &=h(Z_{1},\ldots,Z_{n}) - h(Z_{1},\ldots,Z_{n}|X_{1},\ldots,X_{n})\\
 &\stackrel{\text{(a)}}{=} h(Z_{1},\ldots,Z_{n}) - h(Y_{1},\ldots,Y_{n})\\
&\stackrel{ \text {(b) } }{=} \sum_{i=1}^{n}h\big(Z_{i} | Z_{1}, \ldots, Z_{i-1}\big)-\sum_{i=1}^{n}h\big(Y_{i}\big)\\
 &\stackrel{ \text{(c)}}{\leq} \sum_{i=1}^{n} h\big(Z_{i}) - \sum_{i=1}^{n} h(Y_{i})\\
 &\stackrel{ \text{(d)}}{=} \sum_{i=1}^{n} I(X_{i};Z_{i})\\ &\stackrel{\text{(e)}}= \sum_{i=1}^{n} I(X_{i}/\sigma_{Z_i};Z'_{i}).
 \end{align*}
 Step (a) holds, as $h(Z_{i} | X_{i}) = h(Y_{i})$, (b) holds from the chain rule of differential entropy and the condition that the $Y_{i}$'s are independent random variables,
 (c) holds, as conditioning decreases entropy, (d) holds, as $h\big(Z_{i}) -h(Y_{i})=h(Z_{i})-h(Z_{i}|X_{i})=I(X_{i};Z_{i})$, and (e) holds from the fact that mutual information is scaling invariant. 
As a consequence, 
\begin{align*}
 \lim_{\sigma^{2}_{Y_i}\rightarrow\infty} \sum_{i=1}^{n} I(X_{i};Z_{i})&= \lim_{\sigma_{Z_i}\rightarrow \infty} \sum_{i=1}^{n} I(X_{i}/\sigma_{Z_i};Z'_{i})\\&=\sum_{i=1}^{n}I(0;Z'_{i})=0.
\end{align*}
For the case $W_{i}=X_{i}Y_{i}$, we have
\begin{align*}
  &h(W_{i} | X_{i}) = \int p(X_i) h(W_{i} | X_{i} = x_i) d X_i \\
  &= \int p(x_i) h(X_iY_{i} | X_{i} = x_i) d X_i \\
  &\stackrel{{\text{(a)}}}{=} \int p(x_i) h(Y_{i}) d X_i = h(Y_{i}),
\end{align*}
where (a) holds, because the probability measure of the event $X_i = 0$ is zero. Hence, 
the proof of our second claim follows the same lines as the one presented above, and we conclude that
\begin{align*}
 \lim_{\sigma^{2}_{Y_i}\rightarrow\infty} I(X_{1},&\ldots,X_{n}; W_{1},\ldots,W_{n})\\
 &\leq \lim_{\sigma_{W_i}\rightarrow \infty} \sum_{i \in \mathcal{N}} I(X_{i}/\sigma_{W_i};W'_{i})=0,
\end{align*}
thereby proving our claims. 
\end{proof}
\section{Proof of Proposition~\ref{prop:2}}
\label{ap.2}
\begin{proof}
As $X$ and $Y$ are independent, we have $\sigma^{2}_{Z}=\sigma^{2}_{X}+\sigma^{2}_{Y}$.
For a Gaussian random variable with variance $\sigma^2$, the differential entropy is given by $\frac{1}{2}\log(2\pi e \sigma^2)$, 
so that 
\begin{align}
  I(X;Z)&=h(Z) - h(Y)\nonumber\\
  &\stackrel{{\text{(a)}}}{\leq} \frac{1}{2}\log(2\pi e \nonumber \sigma^2_{Z})-\frac{1}{2}\log(2\pi e \sigma^2_Y) \nonumber\\
  &=\frac{1}{2}\log(1 +\sigma_{X}^2/\sigma_{Y}^2),
\end{align}
where (a) holds, because the maximum entropy of a random variable with fixed variance is achieved by a Gaussian distribution. 
\end{proof}

\bibliographystyle{IEEEbib}
\bibliography{dualpath}

\end{sloppy}

\end{document}